\documentclass[12pt,reqno]{amsart}
\usepackage{graphicx}
\usepackage{amsmath}
\usepackage{amsfonts}
\usepackage{amssymb}
\providecommand{\U}[1]{\protect\rule{.1in}{.1in}}
\newtheorem{theorem}{Theorem}[section]

\newtheorem{corollary}[theorem]{Corollary}

\newtheorem{lemma}[theorem]{Lemma}

\begin{document}

\title{Some Remarks on Cops and Drunk Robbers}

\author{Athanasios Kehagias}
\address{Department of Mathematics, Physics and Computer Sciences, Aristotle University of Thessaloniki, Thessaloniki GR54124, Greece}
\email{\texttt{kehagiat@auth.gr}}

\author{Pawe{\l} Pra{\l}at}
\address{Department of Mathematics, Ryerson University, Toronto, ON, Canada, M5B 2K3}
\email{\texttt{pralat@ryerson.ca}}

\maketitle

\begin{abstract}
The cops and robbers game has been extensively studied under the assumption of optimal play by both the cops and the robbers. In this paper we study the problem in which cops are chasing a \emph{drunk} robber (that is, a robber who performs a random walk) on a graph. Our main goal is to characterize the ``cost of drunkenness.'' Specifically, we study the ratio of expected capture times for the optimal version and the drunk robber one. We also examine the algorithmic side of the problem; that is, how to compute near-optimal search schedules for the cops. Finally, we present a preliminary investigation of the \emph{invisible} robber game and point out differences between this game and \emph{graph search}.
\end{abstract}

\renewcommand{\labelenumi}{(\roman{enumi})} \makeatletter
\renewcommand{\@seccntformat}[1]{\@nameuse{the#1}.\quad} \makeatother

\section{Introduction}\label{sec01}

The game of \emph{Cops and Robbers}, introduced independently by Nowakowski and Winkler~\cite{nw} and Quilliot~\cite{q} almost thirty years ago, is played on a fixed undirected, simple, and finite graph $G$. There are two players, a team of $k$ \emph{cops}, where $k\geq1$ is a fixed integer, and the \emph{robber}. In the first round of the game, the cops occupy any set of $k$ vertices and then  the robber chooses a vertex to start from; in the following rounds, first the cops and then the robber move from vertex to vertex, following the edges of $G$. More than one cop is allowed to occupy a vertex, and the players may remain on their current positions. At every step of the game, both players know the positions of all cops and the robber. The cops win if they capture the robber; that is, if at least one of cop eventually occupies the same vertex as the robber; the robber wins if he can avoid being captured indefinitely. The players are \emph{adversarial}; that is, they play optimally against each other. Since placing a cop on each vertex guarantees that the cops win, we may define the \emph{cop number}, written $c(G)$, to be the minimum number of cops needed to win on $G$. The cop number was introduced by Aigner and Fromme in~\cite{af}.

In this paper we study a new version of the game, in which the robber is \emph{drunk}; that is, he performs a random walk on $G$. The cops are assumed to follow a strategy which is optimal with respect to the robber's random behavior. This version was proposed by D.\ Thilikos during the 4th Workshop on GRAph Searching, Theory and Applications (GRASTA 2011) and he specifically asked the following question: ``what is the \emph{cost of drunkenness}?'' In other words, how much faster than the adversarial robber is the drunk one captured? We try to answer various versions of this question. In addition, we study some algorithmic questions; for example, how to compute the expected capture time for an optimal strategy of cops.

There is a large bibliography on pursuit games on graphs. The reader interested in cops and robbers can start by perusing the surveys~\cite{al,ft,h} and the recent book \cite{NowBook}. To the best of our knowledge, the problem of a drunk robber has not been previously studied in the cops and robbers literature. However there is a strong connection to the \emph{Markov Decision Processes} (MDP) literature; we will comment on this connection (and use it) in Section~\ref{sec06}.  The reader can refer to~\cite{MDPOnline,MDP03,MDP02} for MDP surveys.

While the emphasis of the current paper is on cops chasing the visible robber, we also touch briefly the case of \emph{invisible} robber, both adversarial and drunk. Not much has been written on this problem, but a related problem which has been extensively studied is the \emph{Graph Search} problem, where a team of searchers try to locate in a graph an \emph{invisible} fugitive, who is also assumed to be \emph{arbitrarily fast} and \emph{omniscient} (he always knows the searchers' locations as well as their strategy). A recent comprehensive review of graph search appears in~\cite{ft}.  We emphasize that the graph search problem is similar but \emph{not} identical to cops chasing an invisible robber.

The paper is structured as follows. In Section~\ref{sec02} we present definitions and our notation; the formulation is, naturally, probabilistic. In particular, we define the cost of drunkenness to be the ratio of the capture time for the adversarial robber and the \emph{expected} capture time for the drunk robber. We also present a number of lemmas which we will repeatedly use in the following sections. In Section~\ref{sec03} we obtain bounds on the cost of drunkenness for various special families of graphs; for example, paths, cycles, grids, and complete $d$-ary trees. In Section~\ref{sec04} we look at the problem more generally and show that, for any $c\in[1,\infty)$, there is a graph for which the cost of drunkenness is arbitrarily close to $c$. In Section~\ref{sec06} we connect the cops and drunk robber problem to \emph{Markov Decision Processes} (MDP); that is, Markov chains with a \emph{control input} which can modify the transition probabilities. MDP's provide a natural language for the problem; in particular they are useful in the computation of optimal cop strategies; that is, strategies which minimize the expected robber capture time. We then use the MDP machinery to present algorithms which compute the optimal cop strategy for a given graph and a drunk robber. In Section~\ref{sec07} we give a brief, preliminary discussion of the cost of drunkenness for an \emph{invisible} robber. Finally, in Section~\ref{sec08} we list possible future research directions.

\section{Preliminaries}\label{sec02}

\subsection{Definitions}

Let $G=(V,E)$ be a fixed undirected, simple, and finite graph. Since the game
played on a disconnected graph can be analyzed by investigating each component
separately, we assume that $G$ is connected. We will use the following
notation and assumptions.

\begin{enumerate}

\item There are $k$ cops (for the time being we assume $k\geq c(G)$ but this assumption will be relaxed in later sections).

\item $X_{t}^{i}$ denotes the position of the $i$-th cop at time $t$ ($i \in\{1,2,\ldots,k\}$, $t \in\{0,1,2, \ldots\}$); $X_{t}=(X_{t}^{1},\ldots,X_{t}^{k})$ denotes the vector of all cop positions at time $t$; $\mathbf{X}=( X_{0},X_{1},X_{2},\ldots)$ denotes the positions of all cops during the game ($\mathbf{X}$ may have finite or infinite length).

\item $Y_{t}$ denotes the position of the robber at time $t$ and $\mathbf{Y}=(Y_{0},Y_{1},Y_{2}, \ldots)$ the positions of the robber during the game. (Let us note that there is a correlation between $\mathbf{X}$ and $\mathbf{Y}$; that is, players adjust their strategies observing moves of the opponent.)

\item The moving sequence is as follows: first the cops choose initial positions $X_{0} \in V$, then the robber chooses $Y_{0} \in V$. For $t\in\{1,2,\ldots\}$ first the cops choose $X_{t}$ and then the robber chooses $Y_{t}$. Players use edges of the graph $G$ to move from vertex to another one;  that is, $\{X_t^i, X_{t+1}^i\} \in E$ for $i \in\{1,2,\ldots,k\}$ and $t \in\{0,1,2, \ldots\}$, and $\{Y_t, Y_{t+1}\} \in E$ for $t \in\{0,1,2, \ldots\}$.
 
\item The \emph{capture time }is denoted by $T$ and defined as follows
\[
T=\min\{ t : \exists i \text{ such that }X_{t}^{i}=Y_{t} \};
\]
that is, it is the first time a cop is located at the same vertex as the robber (note that this can happen either after the cops move or after the evader moves). Note that $T<\infty$, since $k\geq c(G)$ and $c(G)$ cops can capture the adversarial robber (and so, of course, the drunk one too).

\end{enumerate}

Assuming for the moment adversarial cops \emph{and} robber, and given initial cop positions $x\in V^{k}$ and robber position $y\in V$, we let $\mathrm{ct}_{x,y}(G,k)=T$. The \emph{$k$-capture time} is defined as follows:
\[
\mathrm{ct}(G,k) =\min_{x\in V^{k}}\max_{y\in V}\mathrm{ct}_{x,y}(G,k).
\]
In other words, we allow our perfect players to choose their initial positions in order to achieve the best outcome. Finally, when $k=c(G)$ we simply write $\mathrm{ct}(G)$ instead of $\mathrm{ct}(G,c(G))$, and call it the \emph{capture time} instead of $c(G)$-capture time. Let us stress one more time that the above quantities are defined under the assumption of \emph{optimal play by both players}.

Next let us assume that the cops are adversarial but the robber is \emph{drunk}. More specifically, we assume the robber performs a random walk on $G$. Given that he is at vertex $v\in V$ at time $t$, he moves to $u\in N(v)$ at time $(t+1)$ with probability equal to $1/|N(v)|$. Note that we do \emph{not} include $v$ in $N(v)$; that is, we consider open, not closed,  neighbourhoods. Moreover, the robber probability distribution does not depend on current position of cops; in particular, it can happen that the robber moves to a vertex occupied by a cop (something the adversarial robber would never do).

Under the above assumptions, the drunk robber game is actually a one-player game and, for given initial configuration and cops strategy, the capture time $T$ is a random variable. For any $x \in V^k$ and $y \in V$, let
\[
\mathrm{dct}_{x,y}\left(  G,k\right)  = \mathbb{E}\left(  T~~|~~X_{0}=x,Y_{0}=y,\text{$k$ cops are used optimally}\right);
\]
in other words, it is the expected capture time given initial cops and robber configurations $x$, $y$ and optimal play by the $k$ cops.

Since the robber is drunk, we cannot expect him to choose the most suitable vertex to start with---instead, he chooses an initial vertex uniformly at random. Cops are, of course, aware of this and so they try to choose an initial configuration so that the expected length of the game is as small as possible. Hence, we define the expected $k$-capture time as follows:
\[
\mathrm{dct}\left(  G,k\right)  =\min_{x\in V^{k}} \sum_{y\in V} \frac{\mathrm{dct}_{x,y}\left(  G,k\right)  }{\left\vert V\right\vert }.
\]
As before, $\mathrm{dct}(G)=\mathrm{dct}(G,c(G))$. We define the \emph{cost of drunkenness} as follows
\[
F(G)=\frac{\mathrm{ct}(G)}{\mathrm{dct}(G)}
\]
and we obviously have $F(G)\geq1$.

While we concentrate on the case $k=c(G)$, it is also natural to consider expected capture time $\mathrm{dct}(G,k)$ for $k\neq c(G)$. The next theorem shows that this is well defined for any $k\geq1$ (in particular, even for $k<c(G)$).

\begin{theorem}\label{thm:dct_finite} 
$\mathrm{dct}(G,k)<\infty$ for any connected graph $G$ and $k\geq1$.
\end{theorem}

\begin{proof}
Let $G=(V,E)$ be any connected graph, $D=D(G)$ be the diameter of $G$, and $\Delta=\Delta(G)$ be the maximum degree of $G$. Fix any vertex $v\in V$, place $k$ cops on $v$, and let $X_{t}^{i}=v$ for all $i$ and $t$ (that is, cops never move; this is clearly a suboptimal strategy). For a given vertex $y\in V$ occupied by the drunk robber, the probability that he uses a shortest path from $y$ to $v$ to move straight to $v$ is at least $(1/\Delta)^{D}$. This implies that, regardless of the current position of the robber at time $t$, the probability that he will be caught after at most $D$ further rounds is at least $\varepsilon= (1/\Delta)^{D}$. Moreover, corresponding events for times $t+iD$, $i\in{\mathbb{N}}\cup\{0\}$ are mutually independent. Thus, we get immediately that
\begin{align*}
\mathbb{E}T  &  =\sum_{t\geq0}\mathbb{P}(T>t)~~\leq~~\sum_{t\geq0} \mathbb{P} \left(  T> \left\lfloor \frac{t}{D} \right\rfloor D \right) \nonumber\\
&  =\sum_{i\geq0} D \cdot\mathbb{P}(T>iD)~~\leq~~D \sum_{i\geq0} (1-\varepsilon)^{i}~~=~~\frac{D}{\varepsilon}~~=~~D \Delta^{D}~~<~~\infty,
\end{align*}
and we are done.
\end{proof}

Let us remark that sharper bounds can be obtained for the capture time of a drunk robber, even in the case that the cops are also drunk; for example see~\cite{Winkler}. However, Theorem \ref{thm:dct_finite} will be sufficient for our needs.

\subsection{Some Useful Lemmas}

We will be using the following version of a well-known Chernoff bound many times so let us state it explicitly.

\begin{lemma}[\cite{JLR}]\label{lem:Chernoff} 
Let $X$ be a random variable that can be expressed as a sum $X=\sum_{i=1}^{n} X_{i}$ of independent random indicator variables where $X_{i}\in\mathrm{Be}(p_{i})$ with (possibly) different $p_{i}=\mathbb{P} (X_{i} = 1)= \mathbb{E}X_{i}$. Then the following holds for $t \ge0$:
\begin{align*}
\mathbb{P} (X \ge\mathbb{E }X + t)  &  \le\exp\left(  - \frac{t^{2}}{2(\mathbb{E }X+t/3)} \right)  ,\\
\mathbb{P} (X \le\mathbb{E }X - t)  &  \le\exp\left(  - \frac{t^{2}}{2\mathbb{E }X} \right)  .
\end{align*}
In particular, if $\varepsilon\le3/2$, then
\begin{align*}
\mathbb{P} (|X - \mathbb{E }X| \ge\varepsilon\mathbb{E }X)  &  \le2 \exp\left(  - \frac{\varepsilon^{2} \mathbb{E }X}{3} \right)  .
\end{align*}
\end{lemma}

Let us now consider the following (simple) random walk on ${\mathbb{Z}}$. Understanding the behaviour of this Markov chain will be important in investigating simple families of graphs later. Let $X_{0}=0$, and for a given $t\geq0$, let
\[
X_{t+1}=
\begin{cases}
X_{t}+1 & \text{with probability }1/2\\
X_{t}-1 & \text{otherwise.}
\end{cases}
\]
It is known that with high probability, random variable $X_{t}$ stays relatively close to zero. We make this precise below using the Chernoff bound.

\begin{lemma}
\label{lem:walk} Let $n \in{\mathbb{N}}$ and $c \in(2, \infty)$. For a simple random walk $(X_{t})$ on ${\mathbb{Z}}$ with $X_{0} = 0$ we have that $|X_{t}|\le c \sqrt{n \log n}$ for every $t \in\{0, 1, \dots, n\}$ with probability at least $1-2n^{1-c^{2}/4}$.
\end{lemma}

\begin{proof}
Fix $n\in{\mathbb{N}}$ and $c\in(2,\infty)$. Let us perform $n$ steps of a simple random walk on ${\mathbb{Z}}$ starting with $X_{0}=0$. Let $Y_{t}$ ($1\leq t\leq n$) denote the number of times the process goes `up' until time $t$. It is clear that $\mathbb{E}Y_{t}=t/2$ and
\[
X_{t}=Y_{t}-(t-Y_{t})=2(Y_{t}-t/2).
\]
For a given $t$, it follows from Chernoff bound (Lemma~\ref{lem:Chernoff}) that
\begin{align*}
\mathbb{P}\left(  X_{t}<-c\sqrt{n\log n}\right)   &  =\mathbb{P}\left( Y_{t}\leq\frac{t}{2}-\frac{c}{2}\sqrt{n\log n}\right) \\
&  \leq\exp\left(  -\frac{(c\sqrt{n\log n}/2)^{2}}{2(t/2)}\right) \\
&  \leq\exp\left(  -\frac{c^{2}}{4}\log n\right)  =n^{-c^{2}/4}.
\end{align*}
A symmetric argument can be used to get that $X_{t}>c\sqrt{n\log n}$ with probability at most $n^{-c^{2}/4}$. Finally, from a union bound we get that the probability that there exists $t$ ($1\leq t\leq n$) with $|X_{t}|>c\sqrt{n\log n}$ is at most $n\cdot2n^{-c^{2}/4}=2n^{1-c^{2}/4}$.
\end{proof}

\section{Bounds on the Cost of Drunkenness}\label{sec03}

In this section we place upper and lower bounds on the cost of drunkenness $F(G)$ when $k$ cops are available. We emphasize the case $k=c(G)$ but also consider values of $k \neq c(G)$. We start with simple graphs (namely: paths, cycles, trees, and grids) in order to prepare for slightly more complicated families in the next section.

\subsection{Paths and a Suboptimal Strategy}

In this subsection we play the game on $P_{n}$, a path on $n$ vertices ($V(P_{n})=\{0,1,\dots,n-1\}$, $E(P_{n})=\{\{i-1,i\}:i\in\{1,2,\dots ,n-1\}\}$). Clearly, $c(P_{n})=1$; that is, one cop can catch the adversarial robber. Since the drunk robber is easier to catch than the adversarial one, let us study the drunk robber playing against a single cop.

In this subsection we will compute the expected capture time using a \emph{suboptimal} strategy, namely starting the cop at $X_0=0$ and moving him to the other end until he reaches $n-1$ (or until capture takes place). It is clear that this strategy achieves capture; furthermore (as will become apparent in the following sections) many optimal strategies can be analyzed using this suboptimal one.

Let $Z_{t}=Y_{t}-X_{t}$ be the distance between players at time $t$. If the drunk robber starts at vertex $k\in\{0,1,\dots,n-1\}$, we have $Z_{0}=Y_{0}=k$. (In order to simplify the argument, we allow players to ``pass each other'' which is never the case in the real game; that is, $Z_{t}$ can be negative.) We can redefine the capture time as
\[
T_n=T_n(k)=\min\{t:Z_{t}\leq0\}.
\]
Now, it is not so difficult to see the behaviour of the sequence $(Z_{t})_{t \ge 0}$. Note that at time t, the maximum distance between players is $n-1-t$ which implies that the robber will be caught in at most $n-1$ steps. We have the following Markov chain to investigate: for $t\in\{0,1,\dots,n-2\}$, if $Z_{t}<n-1-t$, then
\[
Z_{t+1}=
\begin{cases}
Z_{t}-2 & \text{with probability 1/2 (the robber goes toward the cop)}\\
Z_{t} & \text{with probability 1/2 (the robber goes away from the cop)}.
\end{cases}
\]
If $Z_{t}=n-1-t$ (that is, the robber occupies the end of the path), then $Z_{t+1}=Z_{t}-2$ (deterministically).

Consider another Markov chain $Z_{t}^{\prime}$, which has the following simple behaviour: $Z_{0}^{\prime}=k$ and for every $t\geq0$, $Z_{t+1}^{\prime}=Z_{t}^{\prime}-2$ with probability 1/2; otherwise $Z_{t+1}^{\prime}=Z_{t}^{\prime}$. Define $T^{\prime}=\min\{t:Z_{t}\leq0\}$. In other words, we will be chasing the robber on the infinite ray $R$ ($V(R)={\mathbb{N}} \cup\{0\}$, $E(R)=\{\{i-1,i\}:i\in{\mathbb{N}}\}$), which is slightly more difficult for the cop. Hence, it is easy to prove that $\mathbb{E}(T_n~~|~~Z_{0}=k) \leq\mathbb{E}(T^{\prime}~~|~~Z_{0}^{\prime}=k) $. Moreover, it is also easy (using a recursive argument) to show that $\mathbb{E}(T^{\prime}~~|~~Z_{0}^{\prime}=k)=k$, and so $\mathbb{E}(T_n(k)) \leq k$. Now we are ready to show the following.

\begin{theorem}\label{thm:dst_pn} 
Consider that the cop starts on one end of the path $P_{n}$ and moves toward the other end. Let $T_{n}$ be the capture time, provided that the robber is drunk. Then,
\[
\frac{n}{2} \left(  1-O\left(  \frac{\log n}{n}\right)  \right)~~\leq~~\mathbb{E}T_{n}~~\leq~~\frac{n-1}{2}.
\]
\end{theorem}

Before we move to the proof of this theorem let us mention that, in fact, with a slightly more sophisticated argument, it is possible to show that $\mathbb{E}T_n=n/2-O(1)$.

\begin{proof}
Let $n \in{\mathbb{N}}$ and fix any $c>2$. The robber starts his walk on a vertex $k \in\{0, 1, \dots, n-1\}$. Let us note that he is captured after at most $n-1$ steps of the process (deterministically); that is, $T_n(k) \le n-1$. As we already mentioned $\mathbb{E }T_n(k) \le k$. Since the starting vertex for the robber is chosen uniformly at random, we get that $\mathbb{E }T_n \le \sum_{k=0}^{n-1} k/n = (n-1)/2$, so it remains to investigate a lower bound.

Suppose first that $k \le(n-1) - c \sqrt{n \log n}$. It follows from Lemma~\ref{lem:walk} that the robber reaches the other end of the path with probability at most $2n^{1-c^{2}/4}$. If this is the case, we apply a trivial lower bound for $T_n(k)$, namely, $T_n(k) \ge0$; otherwise we get that the (conditional) expectation for $T_n(k)$ is equal to $k$. Hence, $\mathbb{E }T_n(k) \ge k (1-2n^{1-c^{2}/4})$. Suppose now that $k > (n-1) - c \sqrt{n \log n}$. Using Lemma~\ref{lem:walk} one more time, we get that with probability at least $1-2n^{1-c^{2}/4}$ the robber is not caught before time $k - c \sqrt{n \log n}$.

Since the starting vertex for the robber is chosen uniformly at random, we get that
\begin{align*}
\mathbb{E }T_n  &  \ge\frac{1}{n} \sum_{k=0}^{n-1} \mathbb{E }T_n(k)\\
&  \ge\frac{1}{n} \left(  \sum_{k=0}^{n-1-c \sqrt{n \log n}} k + \sum_{k=n-c \sqrt{n \log n}}^{n-1} (k - c \sqrt{n \log n}) \right)  (1-2n^{1-c^{2}/4})\\
&  \ge\left(  \frac{n-1}{2} - c^{2} \log n \right)  (1-2n^{1-c^{2}/4}).
\end{align*}
For a given $n$, the parameter $c$ can be adjusted for the best outcome. To get an asymptotic behaviour, we can use, say, $c=3$ to get that
\[
\mathbb{E }T_n \ge\frac{n}{2} \left(  1 - O \left(  \frac{\log n}{n} \right) \right)  ,
\]
and the proof is complete.
\end{proof}

The proof of the theorem actually gives us more. We get that with probability tending to 1 as $n\rightarrow\infty$, for all starting points for the robber ($k \in\{0, 1, \dots, n-1\}$), the cop needs $k+O(\sqrt{n\log n})$ moves to catch the robber.

\subsection{Paths}

We continue studying a visible robber on $P_{n}$ but we now apply the optimal capture strategy (it is optimal for both adversarial and drunk robber). If $n$ is odd, we start by placing a cop on vertex $(n-1)/2$; if $n$ is even we have two optimal strategies, the cop can start on $n/2$ or $n/2-1$. In any case, after selecting an initial vertex the strategy is the same:\ the cop keeps moving toward the robber. Except for initial placement, this is the strategy examined in the previous subsection and we have $\mathrm{ct}(P_{n})=\lfloor n/2\rfloor$. We easily get the following result.

\begin{theorem}
\[
\frac{n}{4}\left(  1-O\left(  \frac{\log n}{n}\right)  \right)  ~~\leq~~\mathrm{dct(P_{n})~~\leq~~\frac{n}{4}.}
\]
In particular, $\mathrm{dct}(P_{n})=(1+o(1))n/4$ and the cost of drunkenness is
\[
F (P_{n}) = \frac{\mathrm{ct}(P_{n})}{\mathrm{dct}(P_{n})}=2+o(1).
\]
\end{theorem}

\begin{proof}
As we already mentioned, after the robber selects his initial vertex to start from, the game is played essentially on a path of length at most $\lfloor n/2\rfloor+1$. From Theorem~\ref{thm:dst_pn}, we get immediately that
\[
\mathrm{dct}(P_{n})\leq\mathbb{E} T_{\lfloor n/2\rfloor+1} \leq n/4.
\]
For a lower bound, we notice that the length of each subpath is at least $\lfloor n/2\rfloor$. By Theorem~\ref{thm:dst_pn},
\[
\mathrm{dct}(P_{n}) \geq\mathbb{E} T_{\lfloor n/2\rfloor} \geq\frac{n}{4}\left(  1-O\left(  \frac{\log n}{n}\right)  \right)  ,
\]
and the proof is complete.
\end{proof}

In the general case when $k \in\mathbb{N}$ cops are available, we need to `slice' a path into $k$ shorter paths and place a cop on their centers. We get that $\mathrm{dct}(P_{n},k)=(1+o(1))n/(4k)$.

\subsection{Cycles}

Let us play the game on a cycle $C_{n}$ for $n\geq4$ ($V(C_{n})=\{1,2,\dots,n\}$, $E(C_{n})=\{\{i,i+1\}:i\in\{1,2,\dots,n-1\}\}\cup\{\{1,n\}\}$). It is not difficult to see that $c(C_{n})=2$; we use two cops to chase the robber. They start by occupying two vertices at the distance $\lfloor(n+1)/2\rfloor$, the maximum possible distance on the cycle. When the robber selects his vertex to start with, they move toward him and capture occurs at time $\mathrm{ct}(C_{n})=\lfloor(n+1)/4\rfloor$. The same strategy is used when the robber is drunk.

As for paths, one can introduce a random variable $Z_{t}$ to measure the distance between the robber and cops at time $t$. The problem (almost) reduces to the problem on a path. We mention briefly the difference below but the formal proof is omitted. If $n$ is odd, then $Z_{t}$ has exactly the same behaviour as before. However, $Z_{0}=\lfloor(n+1)/2\rfloor$ with probability two times smaller than any other legal starting value (note that a uniform distribution on $V(C_{n})$ is used but there is just one vertex at the distance $\lfloor(n+1)/2\rfloor$). If $n$ is even, then we get a uniform distribution for starting values but the transition from $Z_{t}$ to $Z_{t+1}$ is slightly different, namely, there is a chance for $Z_{t}$ to stay at the same value, provided that the robber occupies the vertex which is at the maximum distance from cops. In any case, it is straightforward to show that both upper and lower bounds still hold so we get the following.

\begin{theorem}\label{thm:dct_cn}
\[
\frac{n}{8}\left(  1-O\left(  \frac{\log n}{n}\right)  \right)  ~~\leq~~\mathrm{dct}(C_{n})~~\leq~~\frac{n+1}{8}.
\]
In particular, $\mathrm{dct}(C_{n})=(1+o(1))n/8$ and the cost of drunkenness is
\[
F(C_{n}) = \frac{\mathrm{ct}(C_{n})}{\mathrm{dct}(C_{n})}=2+o(1).
\]
\end{theorem}

In the general case when $k \in\mathbb{N}$ cops are available, we spread them as evenly as possible. We get that $\mathrm{dct}(C_{n},k)=(1+o(1))n/(4k)$.

\subsection{Trees}

All families of graphs we discussed so far have a very nice property, namely, it is clear what the optimal strategy for the cops is. Once players fix their initial positions (that is, $X_{0}$ and $Y_{0}$), cops must move toward the robber in order to decrease the expected capture time. As we mentioned before, it is natural to measure the distance $Z_{t}$ between players at time $t$; $Z_{t}$ decreases by 2 if the robber makes a bad move or is occupying a leaf; otherwise the distance remains the same. This applies to the family of trees as well (note that $c(T)=1$ for any tree $T$). However, this time it is not clear which vertex should be used for the cop to start with in order to optimize the expected capture time. For this family, the random variable $Z_{t}$ decreases with probability $1/\deg(v)$, provided that the robber occupies vertex $v$, and the behaviour of the sequence $(Z_{t})_{t \ge 0}$ highly depends not only on the degree distribution but on the structure of a tree as well. It is non-trivial to estimate the cost of drunkenness for a particular tree without performing extensive calculations for every vertex as a starting point (these calculations can be performed by computer, using the algorithms of Section \ref{sec0602}). However, some sub-families of trees are still relatively easy to deal with.

Let us consider $d$ regular, rooted tree $T(d,k)$ of depth $k$. The root vertex on the level 0 has $d$ neighbours (children), vertices on levels 1 to $k-1$ have degree $d+1$ (one parent and $d$ children), leaves on the level $k$ have degree 1 (just one parent). There are $d^{i}$ vertices on level $i$ for a total of $(d^{k+1}-1)/(d-1)$ vertices. Due to the symmetry, the cop must start the game on the root. Since the drunk robber prefers to move toward leaves, it is natural to expect that his behaviour is similar to the one of the adversarial robber. Moreover, almost all vertices are located on levels $k-o(k)$ so the robber almost always starts on these vertices which is clearly a good move. We show that the cost of drunkenness is as best as possible; that is, $\mathrm{dct}(T(d,k))$ is tending to $\mathrm{ct}(T(d,k))=k$ as $k \to\infty$.

\begin{theorem}\label{thm:dct_tdk}
\[
k - O(\sqrt{k \log k}) ~~\leq~~\mathrm{dct}(T(d,k))~~\leq~~k.
\]
In particular, $\mathrm{dct}(T(d,k))=(1+o(1))k$ and the cost of drunkenness is
\[
F(T(d,k)) = \frac{\mathrm{ct}(T(d,k))}{\mathrm{dct}(T(d,k))}=1+o(1).
\]
\end{theorem}

\begin{proof}
Suppose that the drunk robber starts on level $i \ge k - \sqrt{k \log k}$. It follows from Lemma~\ref{lem:Chernoff} that with probability $1-O(k^{-1})$ he will be caught on level $k - O(\sqrt{k \log k})$. (In fact, it is also true for $i \ge k/d$, since the robber moves toward leaves with higher rate, namely, with probability $(d-1)/d$. However, an error following from this part is negligible comparing to the other error, so we stay with this obvious bound for $i$.) Therefore,
\begin{align*}
\mathrm{dct}(T(d,k))  &  \ge\sum_{i=k - \sqrt{k \log k}}^{k} \frac{d^{i}}{(d^{k+1}-1)/(d-1)} (k - O(\sqrt{k \log k})) (1 - O(k^{-1})) \\
&  = ( 1 - O(d^{- \sqrt{k \log k}}) ) (k - O(\sqrt{k \log k})) (1 - O(k^{-1})) \\
&  = k - O(\sqrt{k \log k}),
\end{align*}
which finishes the proof.
\end{proof}

\subsection{Grids}

The \emph{Cartesian product} of two graphs $G$ and $H$ is a graph with vertex set $V(G) \times V(H)$ and with the vertices $(u_{1}, v_{1})$ and $(u_{2}, v_{2})$ adjacent if either $u_{1} = u_{2}$ and $v_{1}, v_{2}$ are adjacent in $H$, or $v_{1} = v_{2}$ and $u_{1}, u_{2}$ are adjacent in $G$. We denote the Cartesian product of $G$ and $H$ by $G \square H$. In this subsection, we will study a square grid $P_{n} \square P_{n}$.

It is known that for any two trees $T_{1},T_{2}$, we have $c(T_{1} \square T_{2}) = 2$~\cite{NN}. The capture time of the Cartesian product of trees was recently studied in~\cite{Mehrabian}. It was shown that for any two trees $T_{1}, T_{2}$ we have
\[
\mathrm{ct}(T_{1}\square T_{2}) = \left\lfloor \frac{D(T_{1} \square T_{2})}{2} \right\rfloor = \left\lfloor \frac{D(T_{1}) + D(T_{2})}{2} \right\rfloor,
\]
where $D=D(G)$ is the diameter of $G$. In particular, for a square grid we have that $\mathrm{ct}(P_{n} \square P_{n}) = n-1$.

We will show that the cost of drunkenness for a grid is asymptotic to $8/3$.

\begin{theorem}\label{thm:dct_grid}
\[
\mathrm{dct}(P_{n} \square P_{n})= (1+o(1)) \frac38 n,
\]
and the cost of drunkenness is $F(P_{n} \square P_{n}) = 8/3 + o(1)$.
\end{theorem}

\begin{proof}
Suppose that the drunk robber occupies an internal vertex $(u,v)$. The decision where to go from there can be made in the following way: toss a coin to decide whether modify the first coordinate ($u$) or the second one ($v$); independently, another coin is tossed to decide whether we increase or decrease the value. Hence the robber will move with probability 1/4 to one of the four neighbors of $(u,v)$. Note that, if we restrict ourselves to look at one dimension only (for example, let us call it North/South direction) we see the robber going North with probability 1/4, going South with the same probability and staying in place with probability 1/2. In other words the robber performs a \emph{lazy} random walk on the path. Hence, both coordinates behave similarly to the lazy random walk on integers (move with probability $1/2$; do nothing, otherwise). The same argument as in the previous proofs can be used to show that with probability, say, $1-o(n^{-1})$, the robber stays within the distance $O(\sqrt{n \log n})=o(n)$ from the initial vertex. Hence, if we look at the grid from the `large distance' the drunk robber is not moving at all.

Therefore, since we would like to investigate an asymptotic behaviour, the problem reduces to finding a set $S$ consisting of two vertices such that the average distance to $S$ is as small as possible. Cops should start on $S$ to achieve the best outcome. It is clear that, due to the symmetry of $P_{n} \square P_{n}$, there are two symmetric optimal configurations for set $S$: 
\begin{align*}
S = \{ (n/2+O(1),n/4+O(1)), (n/2+O(1),3n/4+O(1))\}, \\
S = \{ (n/4+O(1),n/2+O(1)), (3n/4+O(1),n/2+O(1))\}.
\end{align*}
In any case, the average distance is
\[
\sum_{u=0}^{n-1} \sum_{v=0}^{n-1} dist( (u,v), S) = (1+o(1)) 8n \int_{x = 0}^{1/2} \int_{y=0}^{1/4} (x+y) dy dx = (1+o(1)) \frac38 n.
\]
The result follows.
\end{proof}

\section{The cost of drunkenness}\label{sec04}

In this section we show that the cost of drunkenness can be arbitrarily close to any real number $c \in[1,\infty)$. In order to do it, we introduce two families of graphs, barbells and lollipops.

\subsection{Barbell}

Let $n \in\mathbb{N}$ and $c \ge0$. The \emph{barbell} $B(n,c)$ is a graph that is obtained from two complete graphs $K_{\lfloor cn \rfloor}$ connected by a path $P_{n}$ (that is, one end of the path belongs to the first clique whereas the other end belongs to the second one). The number of vertices of $B(n,c)$ is $(1+2c)n + O(1)$, $c(B(n,c))=1$. In order to catch (either the adversarial or the drunk) robber, the cop should start at the center of the path and move toward the robber; $\mathrm{ct}(B(n,c)) = n/2 + O(1)$. This family can be used to get any ratio from $(1,2]$.

\begin{theorem}\label{thm:dct_cn1} 
Let $c \ge0$. Then,
\[
\mathrm{dct}(B(n,c))=(1+o(1)) \frac{n}{2} \cdot\frac{1+4c}{2+4c},
\]
and the cost of drunkenness is
\[
F(B(n,c))=\frac{\mathrm{ct}(B(n,c))}{\mathrm{dct}(B(n,c))}=1 + \frac{1}{1+4c}+o(1).
\]
\end{theorem}

\begin{proof}
The drunk robber starts on a clique with probability $(2c)/(1+2c) + o(1)$. If this is the case, the capture occurs at time $n/2 + O(\sqrt{n \log n})$ with probability, say, $1-o(n^{-1})$ by Lemma~\ref{lem:walk}. If the robber chooses a vertex at the distance $k$ from the robber to start with, he is captured after $k+O(\sqrt{n \log n})$ steps, again with probability $1-o(n^{-1})$. Hence the expected capture time is
\[
(1+o(1)) \left(  \frac{2c}{1+2c} \cdot\frac{n}{2} + \frac{1}{1+2c} \cdot
\frac{n}{4} \right)  = (1+o(1)) \frac{n}{2} \cdot\frac{1+4c}{2+4c}.
\]
The theorem holds.
\end{proof}

\subsection{Lollipop}

Let $n \in\mathbb{N}$ and $c \ge0$. The \emph{lollipop} $L(n,c)$ is a graph that is obtained from a complete graph $K_{\lfloor cn \rfloor}$ connected to a path $P_{n}$ (that is, one end of the path belongs to the clique). The number of vertices of $L(n,c)$ is $(1+c)n + O(1)$, and the cop number $c(L(n,c))$ is $1$. In order to catch the perfect robber, the cop should start at the center of the path and move toward the robber; $\mathrm{ct}(L(n,c)) = n/2 + O(1)$. However, it is not clear what the optimal strategy for the drunk robber is. The larger the clique is, the closer to the clique the cop should start the game.

\begin{theorem}\label{thm:dct_cn2} 
Let $c \ge0$. Then,
\[
\mathrm{dct}(L(n,c))=
\begin{cases}
(1+o(1)) \frac{n}{4} \cdot\frac{(\sqrt{2}-1+c)(\sqrt{2}+1-c)}{1+c}, \mbox{ for } c \in[0,1]\\
(1+o(1)) \frac{n}{2(1+c)}, \mbox{ for } c > 1.
\end{cases}
\]
and the cost of drunkenness is
\[
F(L(n,c))=\frac{\mathrm{ct}(L(n,c))}{\mathrm{dct}(L(n,c))}=
\begin{cases}
\frac{2(1+c)}{(\sqrt{2}-1+c)(\sqrt{2}+1-c)} + o(1), \mbox{ for } c \in[0,1]\\ 
(1+c) + o(1), \mbox{ for } c > 1.
\end{cases}
\]
\end{theorem}

Before we move to the proof of this result, let us mention that the cost of drunkenness (as a function of the parameter $c$) has an interesting behaviour. For $c=0$ it is $2$ (we play on the path), but then it is decreasing to hit its minimum of $1+\sqrt{2}/2$ for $c=\sqrt{2}-1$. After that it is increasing back to $2$ for $c=1$, and goes to infinity together with $c$. Therefore, this family can be used to get any ratio at least $1+\sqrt{2}/2 \approx1.71$.

\begin{proof}
Let the cop start on vertex $v$ at the distance $(1+o(1))bn$ from the clique ($b \in[0,1]$ will be chosen to obtain the minimum expected capture time). The drunk robber starts on a clique with probability $c/(1+c)+o(1)$. If this is the case, the capture occurs at time $bn + O(\sqrt{n \log n})$ with probability, say, $1-o(n^{-1})$ by Lemma~\ref{lem:walk}. If the robber chooses vertex at the distance $k$ from the cop, then he is captured, again with probability $1-o(n^{-1})$, after $k+O(\sqrt{n \log n})$ rounds. The robber starts between the cop and the clique with probability $b/(1+c)+o(1)$ and on the other side with remaining probability. Hence the expected capture time is equal to
\begin{align*}
(1+o(1))  &  \left(  \frac{c}{1+c} \cdot bn + \frac{b}{1+c} \cdot\frac{bn}{2} + \frac{1-b}{1+c} \cdot\frac{(1-b)n}{2} \right) \\
&  = (1+o(1)) \frac{n}{1+c} \left(  b^{2} + (c-1)b + 1/2 \right)  .
\end{align*}
The above expression is a function of $b$ (that is, a function of the starting vertex $v$ for the cop) and is minimized at
\[
b = \min\left\{  \frac{1-c}{2} , 0 \right\}  .
\]
The theorem holds.
\end{proof}

It follows immediately from Theorems~\ref{thm:dct_tdk},~\ref{thm:dct_cn1}, and~\ref{thm:dct_cn2} that the cost of drunkenness can be arbitrarily close to any constant $c \ge1$.

\begin{corollary}
For every real constant $c\geq1$, there exists a sequence of graphs $(G_{n})_{n \ge 1}$ such that
\[
\lim_{n\rightarrow\infty}F(G_{n})=\lim_{n\rightarrow\infty}\frac{\mathrm{ct}(G_{n})}{\mathrm{dct}(G_{n})}=c.
\]
\end{corollary}

\section{Computational Aspects}\label{sec06}

In this section we deal with computational aspects of the cop against drunk robber problem. Our analysis holds for any number of cops, that is, we no longer assume that $k=c\left(G\right)$.

\subsection{Computing expected capture time for a given strategy\label{sec0601}}

Suppose that we are given a graph and we fix a strategy before the game actually starts. We will now show how to explicitly compute the probability of capture at time $t\in\{0,1,2,\ldots\}$ as well as the expected capture time.

Fixing a strategy in advance is the best one can do for the invisible robber case (see Section~\ref{sec07}) but for a visible one, cops should adjust their strategy based on the behaviour of the opponent; this will be treated in the next subsection~\ref{sec0602}. However, the approach presented here is less demanding computationally and can be used to provide an upper bound for the optimal expected capture time.

Let $G=(V,E)$ be a connected graph with $V=\{0,1,\ldots,n-1\}$. Letting
\[
P_{i,j}=\Pr\left(  Y_{t}=j~~|~~Y_{t-1}=i\right)
\]
we have
\[
P_{i,j}=
\left\{
\begin{array}
[c]{clc}
\frac{1}{|N(i)|} & \text{for }j\in N(i) & \\
0                & \text{otherwise.} &
\end{array}
\right.
\]
Note that $P$ is the $n\times n$ transition probability matrix governing the robber's random walk in $G$ \emph{in the absence of cops}. To account for capture by the cops, define a new state space $\overline{V}=V\cup\left\{n\right\}$, that is, the old state space augmented by the \emph{capture state }$n$. The corresponding $(n+1)\times(n+1)$ transition matrix is
\[
\overline{P}=
\left(
\begin{array}
[c]{cc}
P          & \mathbf{0}\\
\mathbf{0} & 1
\end{array}
\right).
\]
In the absence of cops, the robber performs a standard random walk on $G$ and never enters the capture state; if however he starts in the capture state, he remains there forever: $\overline{P}_{n,n}=1$. In other words, the Markov chain governed by $\overline{P}$ contains two noncommunicating equivalence classes: $\{0,1,\ldots,n-1\}$ and $\{n\}$.

Suppose now that a single cop is located in vertex $x$. We will denote the corresponding transition probability matrix by $\overline{P} \left(x\right)$. Obviously, $\overline{P}\left(x\right)\neq\overline{P}$. The difference is caused by the possibility of capture, which can occur in two ways.

\begin{enumerate}
\item At the $(t-1)$-th round the robber is located at $x$ and, in the first phase of the $t$-th round, the cop moves into $x$. Then the robber is captured, so $\overline{P}_{x,n}\left(x\right)=1$ and $\overline{P}_{x,y}\left(x\right)=0$ for $y \in V$.

\item At the $(t-1)$-th round the robber is located at $y\neq x$ and, in the second phase of the $t$-th round, he moves from $y$ to $x$. Hence the robber is captured with probability $P_{y,x}$. So, for all $y\in V-\left\{x\right\}$,  $\overline{P}_{y,n}\left(x\right)=P_{y,x}$, $\overline{P}_{y,x}\left(x\right)=0$.
\end{enumerate}

We can summarize the above by writing
\[
\overline{P}\left(x\right)=\left(
\begin{array}
[c]{cc}
P\left(x\right)  & \mathbf{p}(x)\\
\mathbf{0}       & 1
\end{array}
\right),
\]
where $P\left(x\right)$ has 0's in the $x$-th row and column and the corresponding probabilities have been moved into the $\mathbf{p}(x)$ vector. For example, letting $G$ be the path with 5 nodes, the matrices $\overline{P}$ and $\overline{P}\left(2\right)$ are: 
{\tiny
\[
\overline{P}=
\left(
\begin{array}
[c]{cccccc}
0 & 1 & 0 & 0 & 0 & 0\\
1/2 & 0 & 1/2 & 0 & 0 & 0\\
0 & 1/2 & 0 & 1/2 & 0 & 0\\
0 & 0 & 1/2 & 0 & 1/2 & 0\\
0 & 0 & 0 & 1 & 0 & 0\\
0 & 0 & 0 & 0 & 0 & 1
\end{array}
\right),
\quad
\overline{P}\left(2\right)=
\left(
\begin{array}
[c]{cccccc}
0 & 1 & 0 & 0 & 0 & 0\\
1/2 & 0 & 0 & 0 & 0 & 1/2\\
0 & 0 & 0 & 0 & 0 & 1\\
0 & 0 & 0 & 0 & 1/2 & 1/2\\
0 & 0 & 0 & 1 & 0 & 0\\
0 & 0 & 0 & 0 & 0 & 1
\end{array}
\right).
\]
} 
Especially for the placement round of the game ($t=0$)\ we need a different matrix, because the robber does not perform a random-walk, but simply chooses an initial position uniformly at random; if he chooses the one already occupied by the cop, then he is captured immediately. Hence, for this round the appropriate transition matrix is $\widehat{P}\left(x\right)$, which is the unit matrix with the one of the $x$-th row moved to the $\left(n+1\right)$-th column.

Let $\pi_{i}(t)=\mathbb{P}(Y_{t}=i)$ for $i\in\overline{V}$ and $t\in\left\{0,1,\ldots,s\right\}$
and $\pi(t)=\left(\pi_{0}(t),\pi_{1}(t), \ldots,\pi_{n}(t)\right)$; also let
$\widehat{\pi}\left(0\right)  =\left(\frac{1}{n},\frac{1}{n},\ldots,\frac{1}{n},0\right)$.
Then, given a strategy $\mathbf{X}=\left(x_{0},x_{1},\ldots,x_{s}\right)$, the above formulation yields 
\[
\pi\left(0\right)=\widehat{\pi}\left(0\right)\widehat{P}\left(x_{0}\right)
\]
and, for $t\in\left\{1,2, \ldots\right\}$,
\[
\pi\left(t\right)=\pi\left(t-1\right)\overline{P}\left(x_{t-1}\right).
\]
This implies that  $\pi\left(t\right)= \widehat{\pi}\left(0\right)\widehat{P}\left(x_{0}\right)\overline{P}\left(x_{1}\right)\overline{P}\left(x_{2}\right)\ldots\overline{P}\left(x_{t}\right)$. To illustrate this, let us continue the example. Suppose a single cop enters the path and follows the strategy $\mathbf{X}=(0,1,2,3,4)$ (start on one end of the path and move to the other one). Then we have

{\tiny 
$\pi\left(0\right)=\widehat{\pi}\left(0\right)\widehat{P}\left(x_{0}\right)=
\left(
\begin{array}
[c]{cccccc}
1/5 & 1/5 & 1/5 & 1/5 & 1/5 & 0
\end{array}
\right)  
\left(
\begin{array}
[c]{cccccc}
0 & 0 & 0 & 0 & 0 & 1\\
0 & 1 & 0 & 0 & 0 & 0\\
0 & 0 & 1 & 0 & 0 & 0\\
0 & 0 & 0 & 1 & 0 & 0\\
0 & 0 & 0 & 0 & 1 & 0\\
0 & 0 & 0 & 0 & 0 & 1
\end{array}
\right)
=\allowbreak
\left(
\begin{array}
[c]{cccccc}
0 & \frac{1}{5} & \frac{1}{5} & \frac{1}{5} & \frac{1}{5} & \frac{1}{5}
\end{array}
\right)$
}

{\tiny 
$\pi\left(1\right)=\pi\left(  0\right)  \overline{P}\left(x_{1}\right)=
\left(
\begin{array}
[c]{cccccc}
0 & \frac{1}{5} & \frac{1}{5} & \frac{1}{5} & \frac{1}{5} & \frac{1}{5}
\end{array}
\right)
\left(
\begin{array}
[c]{cccccc}
0 & 0 & 0 & 0 & 0 & 1\\
0 & 0 & 0 & 0 & 0 & 1\\
0 & 0 & 0 & 1/2 & 0 & 1/2\\
0 & 0 & 1/2 & 0 & 1/2 & 0\\
0 & 0 & 0 & 1 & 0 & 0\\
0 & 0 & 0 & 0 & 0 & 1
\end{array}
\right)
=\allowbreak
\left(
\begin{array}
[c]{cccccc}
0 & 0 & \frac{1}{10} & \frac{3}{10} & \frac{1}{10} & \frac{1}{2}
\end{array}
\right)$
}

{\tiny 
$\pi\left(2\right)=\pi\left(1\right)\overline{P}\left(x_{2}\right)
=\allowbreak
\left(
\begin{array}
[c]{cccccc}
0 & 0 & \frac{1}{10} & \frac{3}{10} & \frac{1}{10} & \frac{1}{2}
\end{array}
\right)  
\left(
\begin{array}
[c]{cccccc}
0 & 1 & 0 & 0 & 0 & 0\\
1/2 & 0 & 0 & 0 & 0 & 1/2\\
0 & 0 & 0 & 0 & 0 & 1\\
0 & 0 & 0 & 0 & 1/2 & 1/2\\
0 & 0 & 0 & 1 & 0 & 0\\
0 & 0 & 0 & 0 & 0 & 1
\end{array}
\right)
=\allowbreak
\left(
\begin{array}
[c]{cccccc}
0 & 0 & 0 & \frac{1}{10} & \frac{3}{20} & \frac{3}{4}
\end{array}
\right)$
}

{\tiny 
$\pi\left(3\right)=\pi\left(2\right)\overline{P}\left(x_{3}\right)
=\allowbreak
\left(
\begin{array}
[c]{cccccc}
0 & 0 & 0 & \frac{1}{10} & \frac{3}{20} & \frac{3}{4}
\end{array}
\right)
\left(
\begin{array}
[c]{cccccc}
0 & 1 & 0 & 0 & 0 & 0\\
1/2 & 0 & 1/2 & 0 & 0 & 0\\
0 & 1/2 & 0 & 0 & 0 & 1/2\\
0 & 0 & 0 & 0 & 0 & 1\\
0 & 0 & 0 & 0 & 0 & 1\\
0 & 0 & 0 & 0 & 0 & 1
\end{array}
\right)
=\allowbreak
\left(
\begin{array}
[c]{cccccc}
0 & 0 & 0 & 0 & 0 & 1
\end{array}
\right)$
}

The elements $\pi_{n}(t)$ give the probabilities $P(X_{t}=n)$ at time $t$, that is, the probabilities of capture in \emph{at most} $t$ steps. The probabilities of capture \emph{exactly} at time $t$ are then given by $\pi_{n}(t)-\pi_{n}(t-1)$.  The expected capture time (conditional on strategy $\mathbf{X}$ being used) is
\[
\mathbb{E}T=\sum_{t=1}^{\infty}t\cdot\left(\pi_{n}(t)-\pi_{n}(t-1)\right).
\]
In the above example we have
\[
\mathbb{E}T=1\cdot\left(\frac{1}{2}-\frac{1}{5}\right)+2\cdot\left(\frac{3}{4}-\frac{1}{2}\right)+3\cdot\left(1-\frac{3}{4}\right)=\frac{31}{20}.
\]

The approach can be generalized to more than one cop, by letting $\mathbf{x}=(x_{1},x_{2},\ldots,x_{k})$ be a configuration of cops and defining $\overline{P}(\mathbf{x})$, $P(\mathbf{x})$ analogously to the one cop case. Given that the cops follow the strategy $\mathbf{X}=(X_{1},X_{2},\ldots,X_{s})$, the transition probabilities of $Y$ satisfy
\[
\mathbb{P}(Y_{t}=j~~|~~Y_{t-1}=i)=P_{ij}(X_{t})
\]
for $t\leq s$. So the robber process is an inhomogeneous Markov chain, with the transitions controlled by the cops' actions. Markov chains of this type are called \emph{Markov Decision Processes} (MDP) or \emph{Controlled Markov Processes},  where the control function is $X_{t}$; it is a (stochastic) control in the sense that it allows us to change the transition probabilities of $Y_{t}$. We can use the MDP formulation to compute $\mathbb{E}T$ for any given strategy $\mathbf{X}$ in reasonable time. Computing the \emph{optimal} strategy is not computationally viable; for example, with $|V|=n$ and $k$ cops there may exist up to $\Theta((n^{k})^{t})$ strategies of length $t$ (and the same number of corresponding $\mathbb{E}T$'s) to evaluate. In the Section~\ref{sec0602} we will present a computationally viable approach to compute the strategy that is arbitrarily close to the optimal one.

MDP's were introduced in the book \cite{MDP01}; book-length treatments are~\cite{BerTsi,laBarriere,MDP03,MDP02}; an online tutorial is~\cite{MDPOnline}. They have been applied to a version of the cops-robber problem in~\cite{Zadeh}.

\subsection{Computing near-optimal strategies and minimum expected capture time\label{sec0602}}

Let us now present and algorithm to compute $F(G)=\frac{\textrm{ct}(G)}{\textrm{dct}(G)}$ with arbitrarily good precision.  Basically this reduces to computing $\textrm{ct}(G)$ and a good approximation of $\textrm{dct}(G)$, which can be done independently. To this end we present two algorithms, both of which have previously appeared in the literature. To improve the presentation we assign a name to each algorithm and make a few notational modifications; also we point out the similarity between the two algorithms (which apparently has not been noticed before).

\begin{enumerate}
\item The \emph{CAAR} (\emph{C}op \emph{A}gainst \emph{A}dversarial \emph{R}obber) algorithm computes $\text{ct}_{x,y}(G)$ for every initial cop/robber configuration $(x,y)$. In addition, CAAR computes the optimal cop and robber play for every $(x,y)$. Capture time $\text{ct}\left(G\right)$ is easily computed from $\text{ct}(G)=\min_{x}\max_{y}\text{ct}_{x,y}(G)$.

\item Similarly, the \emph{CADR}(\emph{C}op \emph{A}gainst \emph{D}runk \emph{R}obber) algorithm computes (an arbitrarily good approximation of) $\text{dct}_{x,y}(G)$ and the (near-)optimal cop play for every $(x,y)$; drunken capture time $\text{dct}(G)$ is computed from $\text{dct}(G)=\min_{x}\frac{\sum_{y}\text{dct}_{x,y}(G)}{n}$.
\end{enumerate}

CAAR was introduced by Hahn and MacGillivray in~\cite{Hahn2006}. We present the algorithm for the case of a single cop (the generalization for more than one cops is straightforward). Slightly changing notation, we will use $C_{x,y}$ to denote the game duration when the cop is located at $x$, the robber at $y$ and it is the cop's turn to move (in other words, $C_{x,y}$ equals $\text{ct}_{x,y}(G)$). Similarly $R_{x,y}$ denotes game duration when it is the robber's turn to move. For both $C_{x,y}$ and $R_{x,y}$ we assume optimal play by both cop and robber. Let us also define
\[
\widehat{V}^{2}=V\times V-\left\{\left(x,x\right):x\in V\right\},
\]
(that is, $V^{2}$ excluding the diagonal) and for all $x\in V$, let $N^{+}\left(x\right)=N\left(x\right)\cup\left\{x\right\}$ be the closed neighbourhood of $x$. CAAR\ consists of the following recursion (for $i=1,2,\ldots$):
\begin{align}
\forall\left(x,y\right)   &  
\in\widehat{V}^{2}:R_{x,y}^{\left(i\right)}=\max_{y^{\prime}\in N^{+}\left(y\right)}C_{x,y^{\prime}}^{\left(i-1\right)},
\label{eq0601}\\
\forall\left(x,y\right)   &  
\in\widehat{V}^{2}:C_{x,y}^{\left(i\right)}=1+\min_{x^{\prime}\in N^{+}\left(x\right)}R_{x^{\prime},y}^{\left(i\right)}.
\label{eq0602}
\end{align}
$C$ and $R$ are initialized with $C_{x,y}^{(0)}=R_{x,y}^{(0)}=\infty$ for all $x\neq y$. We take $C_{x,x}^{(i)}=R_{x,x}^{(i)}=0$ for $i=0,1,2,\ldots$. Then (\ref{eq0601})-(\ref{eq0602}) is essentially equivalent to the version presented by Hahn and MacGillivray in~\cite{Hahn2006}, with just one difference which we will now discuss.

In (\ref{eq0601})-(\ref{eq0602}) the matrix $C$ is computed iteratively: the ($i-1$)-th matrix $C^{(i-1)}$ is stored and used in the $i$-th iteration to compute $C^{\left(i\right)}$. In numerical analysis this is known as a \emph{Jacobi} iteration. It is well known that an alternative approach to computations of this type is the \emph{Gauss-Seidel} iteration. In this iteration a single copy of $C$ is stored and its elements are updated ``in place.'' In~\cite{Hahn2006}, Hahn and MacGillivray present the Jacobi version of CAAR and prove that the algorithm converges (in a finite number of steps) if and only if $c\left(G\right)=1$. Hence CAAR computes the solution of the equations
\begin{align}
\forall\left(x,y\right)   &  
\in\widehat{V}^{2}:R_{x,y}=\max_{y^{\prime}\in N^{+}\left(y\right)}C_{x,y^{\prime}},
\label{eq0603}\\
\forall\left(x,y\right)&
\in\widehat{V}^{2}:C_{x,y}=1+\min_{x^{\prime}\in N^{+}\left(x\right)}R_{x^{\prime},y},
\label{eq0604}\\
\forall x &
\in V:C_{x,x}=R_{x,x}=0.
\label{eq0604a}
\end{align}
The interpretation of the equations is the following. Equation~(\ref{eq0603}) captures the property that from configuration $\left(x,y\right)$ the robber moves so as to maximize the length of the game; similarly, (\ref{eq0604}) describes the cop's goal to minimize the game duration (since the cop moves in the first phase of each round, 1 time unit must be added to $\min R_{x^{\prime},y}$); finally (\ref{eq0604a}) says that the game ends when cop and robber occupy the same vertex.

Extending the CAAR idea to the \emph{drunk }robber game, let us now use $C_{x,y}$ to denote $\text{dct}_{x,y}(G)$. 
In other words $C_{x,y}$ (respectively, $R_{x,y}$) is the \emph{expected} game duration after the cop's (respectively, robber's) move. Recall (see Subsection \ref{sec0601}) that $P_{y,y^{\prime}}(x)$ is the probability of the robber transiting from $y$ to $y^{\prime}$, given that the cop is at $x$; note that $P(x)$ is a \emph{substochastic} matrix. The analog of (\ref{eq0601})-(\ref{eq0602}) is
\begin{align}
\forall\left(  x,y\right)   &  \in\widehat{V}^{2}:R_{x,y}^{\left(  i\right)
}=\sum_{y^{\prime}\in N\left(  y\right)  }P_{y,y^{\prime}}\left(  x\right)
C_{x,y^{\prime}}^{\left(  i-1\right)  },\label{eq0605}\\
\forall\left(  x,y\right)   &  \in\widehat{V}^{2}:C_{x,y}^{\left(  i\right)
}=1+\min_{x^{\prime}\in N^{+}\left(  x\right)  }R_{x^{\prime},y}^{\left(
i\right)  }\label{eq0606}%
\end{align}
and the analog of (\ref{eq0603})-(\ref{eq0604a}) is
\begin{align}
\forall\left(x,y\right)   &  
\in\widehat{V}^{2}:R\left(x,y\right)=\sum_{y^{\prime}\in N\left(y\right)}P_{y,y^{\prime}}\left(x\right)C_{x,y^{\prime}},
\label{eq0607}\\
\forall\left(x,y\right)   &
\in\widehat{V}^{2}:C_{x,y}=1+\min_{x^{\prime}\in N^{+}\left(x\right)}R_{x^{\prime},y}.
\label{eq0608}\\
\forall x &
\in V:C_{x,x}=R_{x,x}=0.
\label{eq0608a}
\end{align}
We want (\ref{eq0605})-(\ref{eq0606}) to converge to the solution of (\ref{eq0607})-(\ref{eq0608a}). We will discuss convergence conditions (and initialization)\ presently.

Actually (\ref{eq0605})-(\ref{eq0606}) can be simplified. Since the drunk robber does not choose his moves, we can eliminate $R_{x,y}^{(i)}$ from (\ref{eq0605})-(\ref{eq0606}) and obtain the CADR\ algorithm recursion:
\begin{equation}
\forall\left(x,y\right)\in\widehat{V}^{2}:C_{x,y}^{\left(i\right)}
=1+\min_{x^{\prime}\in N^{+}\left(x\right)}
\left(\sum_{y^{\prime}\in N\left(y\right)}P_{y,y^{\prime}}\left(x^{\prime}\right) C_{x^{\prime},y^{\prime}}^{\left(i-1\right)}\right).
\label{eq0609}
\end{equation}

We have derived (\ref{eq0609}) from (\ref{eq0605})-(\ref{eq0606}), which we see as an analog of (\ref{eq0601})-(\ref{eq0602}). However, we will now show that (\ref{eq0609}) is a version of the \emph{value iteration} algorithm, introduced and studied in the MDP\ literature~\cite{BerTsi,laBarriere,MDP03,MDP02}. Consider a general MDP process with state space $S$, action space $A$, transition matrix $Q$ and cost matrix $G(a)$ (that is, $G_{s,s^{\prime}}\left(a\right)$ is the cost of transition $s\rightarrow s^{\prime}$ using action $a$). The state space satisfies $S=S_{T}\cup S_{A}$, where $S_{T}$ are the transient states and $S_{A}$ the absorbing ones; it is assumed that transitions after absorption have zero cost: $G_{s,s^{\prime}}\left(a\right)=0$ for $s,s^{\prime}\in S_{A}$. Let $C_s$ be the expected total cost of the process starting from state $s$ and continuing until absorption. Then \cite{MDP03} $C$ satisfies the equations
\begin{equation}
\forall s\in S_{T}:C_{s}=
\min_{a\in A}\left(G_{s,s^{\prime}}\left(a\right)+\sum_{s^{\prime}\in S_{T}}Q_{s,s^{\prime}}\left(a\right)C_{s^{\prime}}\right)
\label{eq0611}
\end{equation}
and the solutions to (\ref{eq0611}) can be obtained by the following value iteration:
\begin{equation}
\forall s\in S_{T}:C_{s}^{\left(i\right)}=
\min_{a\in A}\left(G_{s,s^{\prime}}\left(a\right)+
\sum_{s^{\prime}\in S_{T}}Q_{s,s^{\prime}}\left(a\right)  C_{s^{\prime}}^{\left(i-1\right)}\right).
\label{eq0610}
\end{equation}
To show that (\ref{eq0610}) can be reduced to (\ref{eq0609}) let us take $S_{T}=\widehat{V}^{2}$ and $A=V$; in other words, states $s=(x,y)$ are cop/robber configurations and actions $a=x^{\prime}$ are new cop positions. Regarding move costs: 
(a) before capture every move has unit cost, 
(b) after capture only moves of the form $(x,x) \rightarrow (x,x)$ 
are possible and these have zero cost; in short
\[
G_{\left(x,y\right),\left(x^{\prime},y^{\prime}\right)}\left(x^{\prime}\right) =
\left\{
\begin{array}
[c]{ll}
1 & \text{if and only if }x\neq y\\
0 & \text{otherwise.}%
\end{array}
\right.
\]
Finally,
\[
Q_{\left(x,y\right),\left(x^{\prime},y^{\prime}\right)}\left(a\right)=
\left\{
\begin{array}
[c]{ll}
P_{y,y^{\prime}}\left(x^{\prime}\right)   & \text{if }a=x^{\prime}\in N^{+}\left(x\right)\text{ and }y^{\prime}\in N\left(y\right)  \\
0 & \text{otherwise.}%
\end{array}
\right.
\]
Using the above, it is easy to reduce (\ref{eq0610}) to (\ref{eq0609}).

The convergence of the CADR algorithm has been studied by several authors, in various degrees of generality~\cite{Zadeh,MDP01,MDP02}.  A simple yet strong result, derived in~\cite{Zadeh}, uses the concept of \emph{proper strategy}: a strategy is called proper if it yields finite expected capture time. It is proved in~\cite{Zadeh} that: if a  proper strategy exists for graph $G$, then the Gauss-Seidel version of CADR converges to the true $C$ for arbitrary $C^{(0)}$ provided $C_{x,y}^{(0)}\geq0$ for all $\left(x,y\right)\in\widehat{V}^{2}$. As we have seen in Theorem~\ref{thm:dct_finite}, the cop has a proper strategy for every $G$.  It can be proved that the Jacobi version of CADR also converges under the same conditions. 

Now, $F(G)$ can be computed, easily. For every pair $(x,y)$, one can obtain a desired approximation of $\text{ct}_{x,y}(G)$ and $\text{dct}_{xy}(G)$ by performing CAAR and CADR, respectively. Then
\[
F\left(G\right)=
\frac{\text{ct}\left(G\right)}{\text{dct}\left(G\right)}=
\frac{\min_{x\in V}\max_{y\in V}\text{ct}_{xy}\left(G\right)}{\min_{x\in V}\frac{1}{\left\vert V\right\vert }\sum_{y\in V}\text{dct}_{xy}\left(G\right)}.
\]
Both CAAR and CADR can be generalized for the case of $k$ cops, replacing $x$ by a $k$-tuple $\mathbf{x}=(x_{1},x_{2},\ldots,x_{k})$; however, execution time of both algorithms increases exponentially with $k$, hence the algorithms are computationally viable only for small $k$'s. Also CADR will work for any transition probability matrix $P$, not just for random walks. Hence, if desired, we can compute the cost of drunkenness for any number of cops (not just for $k=c(G)$) and for non-uniform random walks (i.e., discrete time birth-and-death processes) and other kinds of Markovian robbers.

Both CAAR and CADR can easily provide an optimal and near-optimal cop strategy in \emph{feedback} form $U_{x,y}$, that is, the optimal cop move when the cop/robber configuration is $(x,y)$. This is achieved by recording a minimizing $x^{\prime}$ in (\ref{eq0604}) / (\ref{eq0609}). The optimal robber strategy $W_{x,y}$ (for the adversarial robber) can be similarly obtained by CAAR. For every $(x,y)$ configuration we can have more than one optimal moves, but they all yield the same (optimal) game duration.

We have implemented the CAAR and CADR algorithms in the Matlab package \texttt{CopsRobber}, which can be downloaded from~\cite{web-page}.  We have used this package to perform a number of numerical experiments, some of which are presented in the technical report~\cite{KehPraTR}. This report also contains presentation of the algorithms in pseudo-code and a discussion of various computational issues.

\section{The Invisible Robber\label{sec07}}

In this section we present an introductory discussion of the cops and robber game when the robber is \emph{invisible}; in other words, the cops do not know the robber's location unless he is occupying the same vertex as one of the cops. All the other rules of the game remain the same. This version raises several interesting questions, a full study of which will be undertaken in a future paper. 

Since the cops never see the robber until capture, they cannot use feedback strategies. In other words, the cop strategy is determined before the game starts. This does \emph{not} mean that every cop move is predetermined because in certain cases it makes sense for the cops to randomize their moves. Hence capture time will in general be a random variable, even in the case of adversarial robber (who may also benefit from a randomized strategy).

Let us first examine the case of adversarial invisible robber. It is clear that, given enough cops, expected capture time will be finite. This is obviously true for $| V |$ cops, but in fact $c(G)$ cops suffice, as seen by the following theorem.

\begin{theorem}
\label{thm:random_cops} Suppose that $c(G)$ cops perform a random walk on a connected graph $G$, starting from any initial position. The robber, playing perfectly, is trying to avoid being captured. Let random variable $T$ be the capture time. Then,
\[
\mathbb{E }T < \infty.
\]
\end{theorem}

\begin{proof}
Let $G=(V,E)$ be any connected graph, and let $\Delta=\Delta(G)$ be the maximum degree of $G$. Put $k=c(G)$. For any configuration of cops $x\in V^{k}$ and any vertex occupied by the robber $y\in V$, there exists a winning strategy $S_{x,y}$ that guarantees that the robber is caught after at most $t_{x,y}$ rounds. It is clear that cops will follow $S_{x,y}$ with probability at least $(1/\Delta)^{kt_{x,y}}$. Now, let us define
\[
\varepsilon=\min_{x\in V^{k},y\in V}(1/\Delta)^{kt_{x,y}}=(1/\Delta)^{kT_{0}}>0\text{, where }T_{0}=\max_{x\in V^{k},y\in V}t_{x,y}.
\]
This implies that, regardless of the current position of players at time $t$, the probability that the robber will be caught after at most $T_{0}$ further rounds is at least $\varepsilon$. Moreover, corresponding events for times $t,t+T_{0},t+2T_{0},\dots$ are mutually independent. Thus, we get immediately that
\begin{align}
\mathbb{E}T  &  =\sum_{t\geq0}\mathbb{P}(T>t)~~\leq~~\sum_{t\geq0}
\mathbb{P}\left(T> \left\lfloor \frac{t}{T_{0}} \right\rfloor T_{0}\right)\nonumber\label{eq:1}\\
&  =\sum_{i\geq0}T_{0}\mathbb{P}(T>iT_{0})~~\leq~~T_{0}\sum_{i\geq
0}(1-\varepsilon)^{i}~~=~~\frac{T_{0}}{\varepsilon}~~<~~\infty,
\end{align}
and we are done.
\end{proof}

Hence $c(G)$ is the \emph{minimum} number of cops required to capture the adversarial invisible robber in finite expected time, since this task is at least as hard as capturing the adversarial \emph{visible} robber. Of course, generally it will take longer, comparing to the visible robber case, to capture the invisible robber. Let us define $\mathrm{ict}_{x,y}(G,k)$ to be the expected capture time when the initial cops/robber configuration is $(x,y)$ and both the $k$ cops and the robber play optimally; we also define 
\[
\mathrm{ict}(G,k)=\min_{x\in V^{k}} \max_{y\in V}\mathrm{ict}_{x,y}(G,k)
\]
and, finally, $\mathrm{ict}(G)=\mathrm{ict}(G,c(G))$.

We now turn to the \emph{drunk} invisible robber. He chooses his starting vertex uniformly at random and performs a random walk, as before. For a given starting position $x\in V^{k}$ for $k$ cops, there is a strategy that yields the smallest expected capture time $\mathrm{idct}_{x}(G,k)$. Cops have to minimize this by selecting a good starting position:
\[
\mathrm{idct}(G,k)=\min_{x\in V^{k}}\mathrm{idct}_{x}(G,k).
\]
As usual, $\mathrm{idct(G)=idct(G,c(G))}$ but it makes sense to consider any value of $k\geq1$. The proof of the next theorem is exactly the same as Theorem~\ref{thm:dct_finite} and so is omitted.

\begin{theorem}
$\mathrm{idct(G,k)<\infty}$ for any connected graph $G$ and $k\geq1$.
\end{theorem}

Finally, the cost of drunkenness for the invisible robber game is  $F_i(G)=\frac{\mathrm{ict}(G)}{\mathrm{idct}(G)}$. It follows from last theorem that this graph parameter is well defined (that is, finite).

Let us make a few remarks regarding the invisible robber with ``\emph{infinite}'' speed (actually, what we mean by this is an arbitrarily high speed). Let us define the cop number for this case by $c^{\infty}(G)$; it is the minimum number of cops that have a strategy to obtain a finite expected capture time. It is clear that $c(G)\le c^{\infty}(G)\le s(G)$, where $s(G)$ is the \emph{search number} of $G$, that is, the minimum number of cops required to \emph{clean} the graph in the \emph{Graph Search} (GS) game (mentioned in Section~\ref{sec01}). We want to emphasize that the cops and robber game (with invisible, infinite speed robber) is \emph{different} from the GS game and, in particular, there are graphs for which $c^{\infty}(G)< s(G)$. For example, for the $C_3$ cycle,  $s(C_3)$=2 but $c^{\infty}(G)=1$, namely one cop using a \emph{randomized} strategy,  can capture the invisible, adversary, infinite speed robber in $T$ with $\mathbb{E}T=2$. Similarly, one cop on $K_{1,3}$, the star with 3 rays, can achieve $\mathbb{E}T=11/3$. Many other examples can be found. The main reason for the discrepancy between  $c^{\infty}(G)$ and $s(G)$ is that, in the GS game, the fugitive is assumed \emph{omniscient} and (under one interpretation) this means he knows \emph{in advance} all the cop moves (until the end of the game). In the cops and robber family of games, on the other hand, omniscience is not assumed, either explicitly or implicitly. We can summarize in one phrase: \emph{clearing is harder than capturing} even an infinite speed robber. We intend to further explore this issue, as well the computation of optimal strategies for cops chasing an invisible adversarial robber in a future publication.

We will finish this section with the computation of the cost of drunkenness for two examples (path and cycle) involving an  invisible (unit speed) robber. In both cases the computation is possible because the optimal strategy (for both the cops and the adversarial robber) is ``obvious.'' Our examples are similar to the ones we have considered for the visible robber and proofs are omitted, since they are almost identical to those of  Section~\ref{sec03}. 

Consider the path $P_{n}$ again, with a single cop and an invisible robber. It is clear that the best strategy for the cop (regardless of whether he is playing against a perfect robber or a drunk one) is to start from one end of the path (say, from vertex $0$) and move along the path until the robber is captured. We have $\mathrm{ict}(P_{n})=n-1$. When cops are playing agains a drunk robber, the expected capture time is roughly two times smaller.

\begin{theorem}
\[
\frac{n}{2} \left(  1 - O \left(  \frac{\log n}{n} \right)  \right)  ~~\le~~
\mathrm{idct}(P_{n}) ~~\le~~ \frac{n-1}{2}.
\]
In particular, $\mathrm{idct}(P_{n}) = (1+o(1)) n/2$ and the cost of drunkenness is
\[
F_i(G)=\frac{\mathrm{ict}(P_{n})}{\mathrm{idct}(P_{n})} = 2+o(1).
\]
\end{theorem}

Let us now play the game with two cops and an invisible robber on the cycle $C_{n}$ for $n\geq4$. It is not difficult to see that $s(C_{n})=2=c(C_{n})$. The best cop strategy is to start on vertices $1$ and $n$; the cop occupying vertex $1$ will move toward higher values, the other one will move in the opposite direction. The game ends after $\mathrm{ict}(C_{n})=\lfloor(n-1)/2\rfloor$ steps. When cops are playing against a drunk robber, the expected capture time is roughly two times smaller.

\begin{theorem}
We have
\[
\frac{n}{4}\left(  1-O\left(  \frac{\log n}{n}\right)  \right)  \leq
\mathrm{idct}(C_{n})\leq\frac{n-1}{4}.
\]
In particular, $\mathrm{idct}(C_{n})=(1+o(1))n/4$ and the cost of drunkenness is
\[
F_i(G)=\frac{\mathrm{ict}(C_{n})}{\mathrm{idct}(C_{n})}=2+o(1).
\]
\end{theorem}

\section{Conclusion\label{sec08}}

Most of the results in the paper pertain to the case of a visible (adversarial / drunk) robber, pursued by $k=c(G)$ cops. The cases of arbitrary $k$ and invisible robber have been briefly touched. We conclude the current paper by listing additional questions regarding the cost of drunkennes. We begin by listing several questions related to the visible robber.

\begin{enumerate}
\item Our analysis can be expanded to strategies which use an arbitrary number of cops. As shown in Theorem~\ref{thm:dct_finite}, even a single cop can catch a drunk robber in finite expected time. Hence, for a given $G$ we can study $\mathrm{dct}(G,k)$ as a function of $k$. Obviously this is a decreasing function; what more can be said about it? As a first step in this direction, the numerical approach of Section~\ref{sec06} can be used to explore the properties of $\mathrm{dct}(G,k)$ for a given graph $G$.

\item Let us define $\mathrm{dct}(G,\mathbf{X})$ to be the expected capture time in graph $G$ using strategy $\mathbf{X}$. It is no longer assumed that $\mathbf{X}$ is an optimal strategy. Under what conditions on $\mathbf{X}$ and/or $G$ will $\mathrm{dct}(G,\mathbf{X})$ be finite? Can we use the approach of Section~\ref{sec06} to obtain non-trivial bounds on $\mathrm{dct}(G,\mathbf{X})$?

\item A related question is whether (for a specific $G$ and either optimal or general strategies) expected capture time can be connected to some graph parameter such as treewidth, pathwidth etc.

\item How robust are our results to slight (natural) modifications of the cops/robber game rules? For example, would the cost of drunkenness change if we allowed the robber to loop into its current location (that is, to perform lazy random walk)? What about a ``general'' random walk (that is, with nonuniform transition probabilities). What about \emph{directed} graphs? Finally, does the situation change significantly if the cops and the robber move simultaneously rather than the cops moving first? The algorithm of Section~\ref{sec06} can be easily modified to handle these cases and numerical experiments may be useful for an initial exploration.
\end{enumerate}

One can try to obtain similar results for the \emph{invisible} robber. In Section~\ref{sec07} we showed how our approach can be extended (at least for certain families of graphs) to this case. In the examples we examined (paths, cycles) the optimal cop strategy is obvious. For general graphs, finding the search strategy optimal for the invisible (adversarial / drunk) robber will be more complicated. Is there a (computationally viable, perhaps approximate) algorithm to achieve this?

Finally, let us note that all of the above analyses adopt the cops' point of view. It will be interesting to study the cost of drunkenness for the cops. In other worlds, assuming an adversarial evader and $k$ drunk cops, can we place bounds on the increase of expected capture time as compared to the case of adversarial cops? Theorem~\ref{thm:random_cops} may be used as a starting point to achieve this goal.

\end{document}